\documentclass[runningheads]{llncs}
\usepackage{graphicx}
\usepackage{comment}
\usepackage{cite}
\usepackage{amsmath}
\usepackage{amsfonts}
\usepackage{thm-restate}
\usepackage{hyperref}
\usepackage{subcaption}
\begin{document}
\title{Minimizing the Size of the Uncertainty Regions for Centers of Moving Entities\thanks{This work was partially funded by 
NSERC Discovery Grants
and
the Institute for Computing, Information and Cognitive Systems (ICICS) at UBC.
}}
\titlerunning{Minimizing the Size of the Uncertainty Regions for Centers}
%
\author{William Evans \and
Seyed Ali Tabatabaee}
\authorrunning{W. Evans and SA. Tabatabaee}
%
\institute{Department of Computer Science, University of British Columbia, Vancouver, Canada\\
\email{\{will,salitaba\}@cs.ubc.ca}}
\maketitle              
\begin{abstract}
In this paper, we study the problems of computing the $1$-center, centroid, and $1$-median of objects moving with bounded speed in Euclidean space. We can acquire the exact location of only a constant number of objects (usually one) per unit time, but for every other object, its set of potential locations, called the object's \emph{uncertainty region}, grows subject only to the speed limit.
As a result, the center of the objects may be at several possible locations, called the center's uncertainty region.
For each of these center problems, we design query strategies to minimize the size of the center's uncertainty region and compare its performance to an optimal query strategy that knows the trajectories of the objects, but must still query to reduce their uncertainty.
For the static case of the $1$-center problem in $\mathbb{R}^1$, we show an algorithm that queries four objects per unit time and is $1$-competitive against the optimal algorithm with one query per unit time.
For the general case of the $1$-center problem in $\mathbb{R}^1$, the centroid problem in $\mathbb{R}^d$, and the $1$-median problem in $\mathbb{R}^1$, we prove that the Round-robin scheduling algorithm is the best possible competitive algorithm.
For the center of mass problem in $\mathbb{R}^d$, we provide an $O(\log{n})$-competitive algorithm.
In addition, for the general case of the $1$-center problem in $\mathbb{R}^d$ ($d \geq 2$), we argue that no algorithm can guarantee a bounded competitive ratio against the optimal algorithm.

\end{abstract}
\section{Introduction}
\label{ch:introduction}


Many real-world problems, such as controlling air traffic and providing service to cellular phones, involve moving entities. Therefore, analyzing moving objects has become a topic of interest within the area of theoretical computer science. In many problems, the movement of the objects is unpredictable and data processing must be done in real-time. Moreover, obtaining the exact location of an object at any point in time often entails a cost. Hence, knowing the precise location of all objects at any time is impractical. Instead, for every object, a region of potential locations that we call its uncertainty region is known. Due to the movement of the objects, the size of the uncertainty region for each object grows over time unless its exact location is acquired. The target is to design cost-effective algorithms for highly accurate analysis of moving objects.


Centers have been used to represent a given point set. They have applications in data clustering and facility location. For a given set of points in the Euclidean space, the three most common centers are the following:

\begin{enumerate}
    \item \textbf{$1$-center:}
    Given a set of $n$ points in $\mathbb{R}^d$, the \emph{$k$-center} is a set of $k$ facility locations such that the farthest distance from an input point to its closest facility, is minimized.  For $k=1$, this is the center of the smallest ball that contains the input set.

    \item \textbf{centroid:}
    Given a set of $n$ points in $\mathbb{R}^d$, the \emph{$k$-centroid}, which is the solution to the \emph{$k$-means problem}, is a set of $k$ facility locations that minimize the sum of squared distances between each input point and its closest facility.
    The 1-centroid is called the centroid.
    Given a set of $n$ weighted points in $\mathbb{R}^d$, the \emph{center of mass} is the average position of all points, weighted according to their masses.
    When all input points have unit weights, the center of mass is the centroid.
    
    \item \textbf{$1$-median:}
    Given a set of $n$ points in $\mathbb{R}^d$, the $k$-median is the set of $k$ facilities such that the average distance from an input point to its closest facility is minimized.
\end{enumerate}

\subsection{Model and Definitions}
\label{ch:model}

Given a center function and a set of the initial locations of $n$ moving objects ($n \geq 2$) in $\mathbb{R}^d$ 
where the speed of every object is bounded by $v$, we consider the problem of finding query strategies
(which can query once at the end of each unit of time) 
for minimizing a measure which is the maximum size of the uncertainty region for the specified center
over all query times.
We define the size of a region as the maximum pairwise distance between the points of that region.
We only care about the uncertainty regions at the times when a query has just been made and one object has an uncertainty region of size $0$.
Considering that for some instances of the problem the 
measure is large for every query strategy, we analyze the performance of our algorithms in a competitive framework,
where the competitive factor is the ratio of our algorithm's measure to that of the optimal algorithm. We assume that the optimal algorithm knows the trajectories of the objects, however it must still query to keep the uncertainty regions for the moving objects, and hence the uncertainty region of their center, small.

In some parts of this paper, we consider the weighted version of the problem where each object has a positive weight. However, we mostly discuss the unweighted version of the problem; hence, we assume that all objects have unit weights unless stated otherwise.
Figure~\ref{fig:com} illustrates the weighted version of the centroid problem (the center of mass problem) for a small set of weighted objects with uncertain locations.
We also consider a version of the problem where objects have different maximum speeds. But, we assume that all objects have the same maximum speed unless stated otherwise.
We define the static case of our problem as a special case where although objects have a maximum speed of $v$ and their uncertainty regions grow accordingly, they are actually static (they do not move).
To achieve better and more meaningful competitive ratios, we may allow our algorithms to query 
more than 
one object
per unit time.

\begin{figure}[ht]
    \centering
    \includegraphics{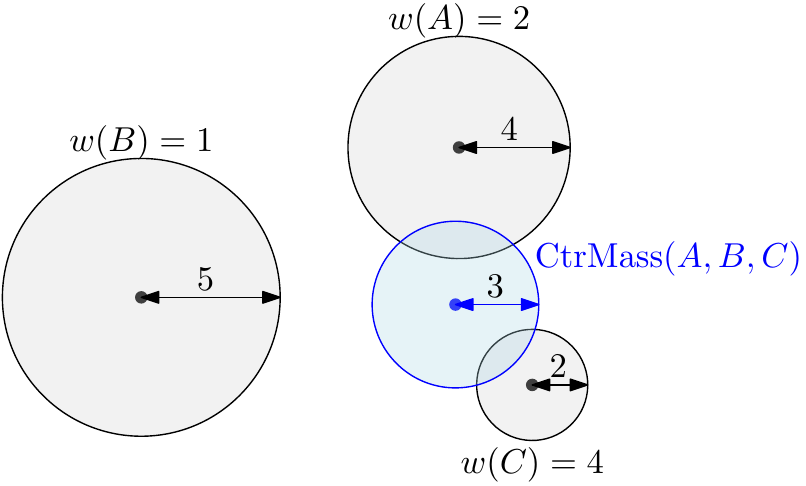}
    \caption{The uncertainty regions for three weighted objects $A$, $B$, and $C$, along with 
    the potential locations (uncertainty region) of
    their center of mass.}
    \label{fig:com}
\end{figure}

\subsection{Contribution and Organization}

In this paper, we study the aforementioned center problems with moving entities in the Euclidean space. We investigate the problem of designing competitive query strategies for minimizing the size of the uncertainty region of these different centers.
%
For the static case of the $1$-center problem in $\mathbb{R}^1$, we show that there exists an algorithm that queries a constant number of objects per unit time and is $1$-competitive against the optimal algorithm with one query per unit time. 
For the $1$-center problem with moving entities in $\mathbb{R}^1$, we provide an algorithm with the best possible competitive ratio of $O(n)$. For this problem in $\mathbb{R}^d$ ($d \geq 2$), we show that no algorithm can guarantee a bounded competitive ratio against the optimal algorithm.
For the centroid problem with moving entities in $\mathbb{R}^d$ ($d \geq 1$), we present an algorithm that is $1$-competitive against the optimal algorithm.
For the weighted version of the centroid problem (the center of mass problem) in $\mathbb{R}^d$ ($d \geq 1$) where objects have different maximum speeds, we provide an $O(\log{n})$-competitive algorithm.
For the $1$-median problem with moving entities in $\mathbb{R}^1$, we give an algorithm with the best possible competitive ratio of $O(n)$.
Several of our algorithms use \emph{Round-robin} scheduling for a subset of the entities, which repeatedly queries the entities in the subset in a fixed order.


The rest of this paper is organized as follows. Section~\ref{ch:background} provides background information on the introduced center problems and computing functions with uncertain inputs and moving data. Sections~\ref{ch:1center}, \ref{ch:centerofmass}, and \ref{ch:1median} study the $1$-center problem, the centroid problem, and the $1$-median problem, respectively, in the presented model. Finally, Section~\ref{ch:conclusion} concludes the results presented in this paper and lists some open problems.
\section{Background}
\label{ch:background}

The $k$-center problem and its special case, the $1$-center problem (also known as the minimum enclosing ball problem) have various applications, such as data classification and facility location. The $k$-center problem is known to be NP-hard in the Euclidean plane~\cite{megiddo1984complexity}. 
However, polynomial-time 2-approximation algorithms exist for the problem in any metric space~\cite{hochbaum1985best, feder1988optimal}.
The Euclidean $1$-center problem always has a unique solution.
Exact polynomial-time algorithms~\cite{chrystal1885problem, welzl1991smallest} exist for the problem in fixed dimensions. When the dimension is not fixed, polynomial-time approximation algorithms with a factor of $1 + \epsilon$ exist~\cite{yildirim2008two}. The speed of the Euclidean $1$-center of moving objects with bounded speed is unbounded in $\mathbb{R}^d$ ($d \geq 2$)~\cite{bereg2006competitive}.
Durocher~\cite{durocher2006geometric} 
provided bounded-velocity approximations for the Euclidean $k$-center of moving objects with bounded speed when $k \leq 2$
, and showed
that no bounded-velocity approximation is possible for the problem when $k \geq 3$.

The $k$-means problem
has applications in data clustering and facility location. This problem has been proved to be NP-hard in $\mathbb{R}^d$ even for $k = 2$~\cite{drineas2004clustering}. Lloyd's method~\cite{lloyd1982least} has been widely used to find a local minimum for the objective function of the $k$-means problem. This method starts with an arbitrary $k$-clustering of input points and computes the centroids of the clusters. Then the method repeatedly assigns each point to its closest cluster center and recomputes the cluster centers.
Constant factor approximation algorithms have been proposed for the Euclidean $k$-means problem~\cite{kanungo2002local, ahmadian2019better}. The best known approximation guarantee for this problem is 6.357~\cite{ahmadian2019better}.
The center of mass is a unique point
where the aggregate mass of a set of objects is concentrated.
The center of mass problem has applications in various fields including physics, astronomy, and engineering.
The speed of the center of mass is the sum of each object's momentum divided by the total weight of all objects. Hence, this speed is bounded by the maximum speed that objects can have.

The $k$-median problem is another well-studied problem with applications in data clustering and facility location. This problem is known to be NP-hard in the Euclidean plane~\cite{megiddo1984complexity}. Hence, constant factor approximation algorithms have been proposed for this problem~\cite{charikar1999constant, ahmadian2019better}. The best known approximation ratio for this problem is 2.633~\cite{ahmadian2019better}.
The $1$-median problem, also known as the Fermat-Weber problem~\cite{drezner2004facility}, is a special case of the $k$-median problem.
The Euclidean $1$-median is unique when the input points are not collinear \cite{kupitz1997geometric} or when the number of input points is odd. However, if the number of input points is even and the points are collinear, then any point that lies on the line segment between the two middle input points is a $1$-median of the points. By convention, the Euclidean $1$-median of such a set of points is defined as the midpoint of the two middle points in the set.
In general, the exact position of the Euclidean $1$-median cannot be calculated using radicals over the field of rationals when the number of points is greater than or equal to five~\cite{bajaj1988algebraic}. Consequently, $(1 + \epsilon)$-approximate solutions have been designed for this problem~\cite{badoiu2002approximate, cohen2016geometric}.
The Euclidean $1$-median moves discontinuously in $\mathbb{R}^d$ ($d \geq 2$)~\cite{durocher2006geometric}.
Durocher~\cite{durocher2006geometric} 
provided bounded-velocity approximations for the Euclidean $1$-median of moving objects with bounded speed, and showed
that no bounded-velocity approximation is possible for the Euclidean $k$-median problem when $k \geq 2$.

Computing functions with uncertain inputs has been the subject of multiple research projects~\cite{feder2000computing, khanna2001computing, suyadi2012computing}. Geometric problems have also been studied with uncertainty and methods such as witness algorithms have been proposed to address those problems~\cite{bruce2005efficient}. Due to the unpredictability of the movements in many problems that involve moving objects, the real-time processing of moving data is closely related to computing with uncertainty. Kahan~\cite{kahan1991real, kahan1991model} studied the maximum problem, the sorting problem, and some geometric problems with moving data, where the target was to reduce data acquisition costs. Furthermore, competitive query strategies have been developed for problems where the number of queries per unit time is bounded~\cite{evans2013competitive, evans2014query, zheng2020scheduling}.
\section{The 1-Center Problem}
\label{ch:1center}

In this section, we consider the $1$-center problem with moving entities.
We first consider a \emph{static} version of the problem
where all queries happen to return the object's original location,
though we must query to confirm this.
Solutions in this case are simpler than in the \emph{general} case, in which the movement of objects may change their importance in calculating the center, however the static case is still challenging.  

\subsection{The Static Case in \texorpdfstring{$\mathbb{R}^1$}{R\^1}}
\label{sc:1cstatic}

For the static case of the $1$-center problem in $\mathbb{R}^1$, we prove the existence of an algorithm that queries four objects per unit time and is $1$-competitive against the optimal algorithm with one query per unit time.

We begin by introducing a lemma which we will use to prove the existence of an algorithm with the aforementioned guarantee.
This lemma considers a special case of the problem of windows scheduling without migration \cite{bar2007windows} and entails the notion of pinwheel scheduling~\cite{holte1989pinwheel, chan1992general}. Given a multiset of positive integers $A = \{a_1, ..., a_n\}$, the pinwheel scheduling problem asks for an infinite sequence over $\{1, ..., n\}$ such that any integer $i \in \{1, ..., n\}$ appears at least once in any $a_i$ consecutive entries of the sequence. The density of $A$ is defined as $d(A) = \sum_{i=1}^{n}{\frac{1}{a_i}}$. A necessary condition for schedulability is $d(A) \leq 1$~\cite{holte1989pinwheel}. Moreover, $d(A) \leq 0.75$ is a sufficient condition for schedulability~\cite{fishburn2002pinwheel}.
Based on these results, we prove the lemma in Appendix~\ref{app:sched}.

\begin{lemma} \label{lem:sched}
Given a multiset of positive integers $A$ where $d(A) \leq 2$, it is possible to partition $A$ into three subsets, which are themselves multisets, such that each subset has a pinwheel schedule.
\end{lemma}

Windows scheduling without migration is a restricted version of the windows scheduling problem \cite{bar2003windows}. 
Given a multiset of positive integers $A = \{a_1, ..., a_n\}$ and a positive integer $h$, the windows scheduling problem asks whether it is possible to schedule $n$ pages on $h$ channels, with at most one page on each channel at any time, such that the time between two consecutive appearances (on any channel) of the $i$-th page is at most $a_i$ ($1 \leq i \leq n$).
In windows scheduling without migration, all appearances of a page must be on the same channel.
Allowing migration would be fine for
our use of Lemma~\ref{lem:sched} in proving the following Theorem~\ref{thm:static}, even though our proof of the lemma does not make use of this flexibility.
It has been shown that windows scheduling is possible on $d(A) + O(\ln{(d(A))})$ channels \cite{bar2003windows}.
Nevertheless, Lemma~\ref{lem:sched} provides a better guarantee for the special case that it considers.


Now, we are ready to prove the main theorem in this subsection.

\begin{theorem} \label{thm:static}
For the static case of the $1$-center problem in $\mathbb{R}^1$, there exists an algorithm that queries four objects per unit time and is $1$-competitive against the optimal algorithm with one query per unit time.
\end{theorem}

\begin{proof}
Without loss of generality, we assume that the maximum speed $v$ is unit.
We sort the objects based on their positions and let $x_i$ denote the position of the $i$-th object ($x_1 \leq ... \leq x_n$).
We define $f(x,y) = \lfloor |y-x| \rfloor + 1$. To prevent the uncertainty region of a static object at position $x$ from going beyond a point at position $y$, we need to acquire the exact location of that object at least once in every $f(x, y)$ queries.
We let $b$ denote the smallest positive value such that $\sum_{i=1}^{n}{\frac{1}{\min(f(x_i, x_n + b), f(x_i, x_1 - b))}} \leq 1$.
Hence, considering the necessary condition for schedulability~\cite{holte1989pinwheel}, no query strategy can ensure that the uncertainty region of each object stays within $[x_1 - b', x_n + b']$ for any value $b' < b$.

The size of the uncertainty region for the $1$-center in $\mathbb{R}^1$ is equivalent to the average of the size of the uncertainty region for the maximum and the size of the uncertainty region for the minimum. Given a set of points in $\mathbb{R}^1$, the maximum problem asks for the position of the maximum point and the minimum problem asks for the position of the minimum point in the set.
Regardless of the query strategy, at some point in the future, either the uncertainty region for the maximum will include $[x_n, x_n + b]$, or the uncertainty region for the minimum will include $[x_1 - b, x_1]$.
If $|x_n - x_1| \leq 1$, the union of the uncertainty regions for the maximum and the minimum will always include $[x_1, x_n]$. Otherwise, the union of the two uncertainty regions will always intersect $[x_1, x_1 + 1] \cup [x_n - 1, x_n]$ such that the total size of the intersection is at least $1$.
Consequently, regardless of the query strategy, the size of the uncertainty region for the $1$-center will be at least
$\frac{b + \min (|x_n - x_1|, 1)}{2}$
at some point in the future. For the optimal algorithm, this is a lower bound on the maximum size of the uncertainty region for the $1$-center.

We now explain how to use four queries per unit time and be $1$-competitive against the optimal algorithm with one query per unit time.
First, we show how to compute $b$.
We know that $0 \leq b \leq n$ (because the Round-robin algorithm can ensure that the uncertainty region of each object stays within $[x_1 - n, x_n + n]$). Furthermore, there exists an index $1 \leq i \leq n$ such that either $\lfloor x_i - x_1 + b \rfloor = x_i - x_1 + b$ or $\lfloor x_n + b - x_i \rfloor = x_n + b - x_i$ (otherwise, there exists a query strategy that can ensure that the uncertainty region of each object stays within $[x_1 - b', x_n + b']$, for some $b' < b$). Therefore, there are $O(n^2)$ possible values for $b$ and we can find $b$ using a binary search.
Further, we show that $f(x_i, x_n + b) \leq 2  f(x_i, x_n + \frac{b}{2})$. We have
\begin{align*}
    f(x_i, x_n + b) = & \ \lfloor x_n + b - x_i \rfloor + 1 \\
    \leq & \ \lfloor x_n + \frac{b}{2} - x_i \rfloor + \lfloor \frac{b}{2} \rfloor + 2 \\
    \leq & \ 2  \lfloor x_n + \frac{b}{2} - x_i \rfloor + 2 \\
    = & \ 2  f(x_i, x_n + \frac{b}{2}).
\end{align*}
Similarly, we have $f(x_i, x_1 - b) \leq 2  f(x_i, x_1 - \frac{b}{2})$. Hence, we deduce that
\[\sum_{i=1}^{n}{\frac{1}{\min(f(x_i, x_n + \frac{b}{2}), f(x_i, x_1 - \frac{b}{2}))}} \leq  \sum_{i=1}^{n}{\frac{2}{\min(f(x_i, x_n + b), f(x_i, x_1 - b))}} \leq 2.\]
Therefore, by Lemma~\ref{lem:sched}, we can use three queries per unit time to maintain the uncertainty regions of all objects within $[x_1 - \frac{b}{2}, x_n + \frac{b}{2}]$. We let the fourth query repeatedly switch between the object at position $x_1$ and the object at position $x_n$. This way, the total size of the intersection between $[x_1, x_n]$ and the union of the uncertainty regions for the maximum and the minimum will be at most $\min (|x_n - x_1|, 1)$. Consequently, using this algorithm, the size of the uncertainty region for the $1$-center will be at most 
$\frac{b + \min (|x_n - x_1|, 1)}{2}$
at any point in the future. Hence, the presented algorithm that queries four objects per unit time is $1$-competitive against the optimal algorithm with one query per unit time.
\qed
\end{proof}

\subsection{The General Case in \texorpdfstring{$\mathbb{R}^1$}{R\^1}}
\label{sc:1c1d}

For the general case of the $1$-center problem in $\mathbb{R}^1$, we prove that the Round-robin algorithm for querying objects 
achieves the best possible competitive ratio 
against the optimal algorithm.

\begin{theorem}
For the $1$-center problem in $\mathbb{R}^1$, the Round-robin scheduling algorithm keeps the maximum size of the uncertainty region within $O(v  n)$ at any point in the future and achieves the best possible competitive ratio of $O(n)$ against the optimal algorithm.
\end{theorem}

\begin{proof}
The Round-robin algorithm keeps the size of the uncertainty region of each object within $O(v  n)$ because it acquires the exact location of each object once in every $n$ queries. For $1 \leq i \leq n$, let $[s_i, e_i]$ denote the uncertainty region of the $i$-th object ($s_i$ is the starting point and $e_i$ is the ending point of the uncertainty region). At any time, let $m$ be the index of the object with the highest ending point ($e_m = \max_{1 \leq i \leq n}{e_i}$). Thus, $[s_m, e_m]$ includes the uncertainty region for the maximum. Consequently, the size of the uncertainty region for the maximum is less than or equal to $|e_m - s_m|$ which is bounded by $O(v  n)$. Similarly, the size of the uncertainty region for the minimum is bounded by $O(v  n)$ at any point in the future. Therefore, using the Round-robin algorithm, the size of the uncertainty region for the $1$-center (which is equivalent to the average of the size of the uncertainty region for the maximum and the size of the uncertainty region for the minimum) will never exceed $O(v  n)$.

For the optimal algorithm, at any point after the first query, between the object with the maximum last known exact position and the object with the minimum last known exact position, at least one will have not been queried last; thus, the size of the uncertainty region of that object will be $\Omega(v)$. Consequently, the size of the uncertainty region for at least one of the maximum or minimum will be $\Omega(v)$. Hence, the size of the uncertainty region for the $1$-center will be $\Omega(v)$ at any point after the first query. Considering that the Round-robin algorithm keeps the maximum size of the uncertainty region within $O(v  n)$, it achieves a competitive ratio of $O(n)$ against the optimal algorithm.

We show that $O(n)$ is the best possible competitive ratio against the optimal algorithm. We consider an example with $n$ objects initially located at the origin. One of those objects moves with a speed of $v$ in the positive direction and another one moves with a speed of $v$ in the negative direction. The rest of the objects do not move. In this example, the optimal algorithm acquires the exact location of each of the two moving objects once in every two queries and maintains the size of the uncertainty region for the $1$-center within $O(v)$.
However, any algorithm that does not know the future object trajectories may fail to query the two moving objects in its first $n-2$ queries;
hence, the size of the uncertainty region for the $1$-center can become $\Omega(v  n)$.
This is true even if we allow the algorithm to query a constant number of objects (instead of one) per unit time. For this reason, $O(n)$ is the best possible competitive ratio.
\qed
\end{proof}

The Round-robin algorithm keeps the maximum size of the uncertainty region for the $1$-center within $O(v  n)$, and so does the optimal algorithm. We show that even for the optimal algorithm, the maximum size of the uncertainty region can be $\Omega(v  n)$.
The proof is presented in Appendix~\ref{app:1c1dOmega}.

\begin{proposition} \label{prop:1c1dOmega}
The maximum size of the uncertainty region for the $1$-center can be $\Omega(v  n)$ for the optimal algorithm.
\end{proposition}

We now provide an upper bound on the competitive ratio of the Round-robin algorithm for the $1$-center problem where objects have different maximum speeds.
The proof is presented in Appendix~\ref{app:1cdiffspeed}.

\begin{proposition} \label{prop:1cdiffspeed}
For the $1$-center problem in $\mathbb{R}^1$ where objects have different maximum speeds, the Round-robin scheduling algorithm achieves a competitive ratio of $O(v_M  n / v_m)$ against the optimal algorithm, where $v_M$ is the highest maximum speed and $v_m$ is the lowest maximum speed.
\end{proposition}

\subsection{The General Case in \texorpdfstring{$\mathbb{R}^d$}{R\^d}}
\label{sc:1cdd}

For the general case of the $1$-center problem in $\mathbb{R}^d$ ($d \geq 2$), we show that the maximum size of the uncertainty region can be unbounded for the optimal algorithm.
The proof of the following proposition is inspired by the proof of Theorem 2 of \cite{bereg2006competitive} and can be found in Appendix~\ref{app:1csizeunb}.

\begin{proposition} \label{prop:1csizeunb}
The maximum size of the uncertainty region of the $1$-center in $\mathbb{R}^d$ ($d \geq 2$) can be unbounded for the optimal algorithm.
\end{proposition}

Next, we show that it is impossible to find an algorithm that works well compared to the optimal algorithm.

\begin{theorem} \label{thm:1crd}
The maximum size of the uncertainty region of the $1$-center of points in $\mathbb{R}^d$ ($d \geq 2$) for any
algorithm that does not know the trajectories may be an arbitrary factor larger than that obtained by the
optimal algorithm.
\end{theorem}

\begin{proof}
We provide an example for which we prove that no algorithm that does not know the trajectories of the objects can guarantee a bounded competitive ratio against the optimal algorithm. Although we explain the example in $\mathbb{R}^2$ here, we could consider the same example in more than two dimensions, and the proof would be very similar.
Let $\mathbf{a} = (0, G)$, $\mathbf{b} = (-2L, 0)$, $\mathbf{c} = (2L, 0)$, $\mathbf{u} = \mathbf{c} - \mathbf{a}$, and $\mathbf{e} = \mathbf{c} + v (n-3) \mathbf{\hat{u}}$, where $n \geq 28$, $L$ is arbitrarily large (much larger than $n$ and $v$), and $G$ is much larger than $L$ in a way that for $y \geq \frac{G}{3}$, we have
$\sqrt{100 L^2 + y^2} \leq y + v$.
We consider an example with $n$ objects where one object (denoted by $A$) is always located at $\mathbf{a}$, one object (denoted by $B$) is always located at $\mathbf{b}$, one object (denoted by $E$) starts at $\mathbf{e}$ and always moves with a maximum speed of $v$ in the direction of $\mathbf{\hat{u}}$, and $n-3$ objects start at $\mathbf{e}$ and move with a maximum speed of $v$ in the direction of $-\mathbf{\hat{u}}$ until they stop at $\mathbf{c}$.

For the above example, the actual $1$-center is always the midpoint of 
the positions of $A$ and $E$.
Consequently, the x-coordinate of the actual $1$-center is always greater than $L$.
Nevertheless, any algorithm that does not know the future object trajectories may fail to query $E$ in its first $n-3$ queries.
Thus, for any such algorithm after the first $n-3$ queries, $\mathbf{c}$ is a potential location of each of the $n-2$ objects that were initially at $\mathbf{e}$, $\mathbf{a}$ is a potential location of $A$, and $\mathbf{b}$ is a potential location of $B$. Therefore, the $1$-center could potentially be located on the line $x = 0$. Consequently, for any algorithm that does not know the future object trajectories, the size of the uncertainty region for the $1$-center can become greater than $L$ (which is arbitrarily large), after the $(n-3)$-th query.

Now, we only need to prove that the optimal algorithm can keep the maximum size of the uncertainty region for the $1$-center within a function of $n$ and $v$ for the example mentioned above. We consider an algorithm that queries each of $A$, $B$, and $E$ once in every four queries and each other object once in every $4  (n-3)$ queries. Hence, the sizes of the uncertainty regions of $A$, $B$, and $E$ will never exceed $6v$. Also, the size of the uncertainty region of each of the other objects will always be $O(v  n)$. We prove that using this algorithm, the maximum size of the uncertainty region for the $1$-center is at most some function of $n$ and $v$ for the provided example. This will prove the same for the optimal algorithm which works at least as well as the aforementioned algorithm.

Let us disregard $B$ for now and analyze the uncertainty region for the $1$-center of the other objects. 
For those objects, the radius of the minimum enclosing circle 
is at most half of the distance between the two farthest points in the uncertainty regions of $A$ and $E$. This is because the line segment between those points is a diameter of a circle that contains the uncertainty regions of all objects.
Therefore, after $q$ queries ($q \geq 0$) made by the algorithm described above, the radius of the minimum enclosing circle is at most $r_q = \frac{1}{2} (3v + |\mathbf{e} - \mathbf{a}| +qv) = \frac{1}{2} |\mathbf{c} - \mathbf{a}| + \frac{v}{2} (n + q)$. Considering that the distance between the $1$-center and some point in the uncertainty region of $A$ must be at most $r_q$ after the $q$-th query, the $1$-center is within distance $r_q + 3v$ from $\mathbf{a}$.
Furthermore, the distance between the $1$-center and some point in the uncertainty region of $E$ must be at most $r_q$ after the $q$-th query. Thus, the distance between the $1$-center and $\mathbf{e} + qv \mathbf{\hat{u}}$ is less than or equal to $r_q + 6v$. Consequently, the distance between the $1$-center and $\mathbf{a}$ is at least $|\mathbf{e} - \mathbf{a}| +qv - r_q - 6v = r_q - 9v$ after the $q$-th query.

The $1$-center of a set of points in the Euclidean space is within the convex hull of those points (otherwise, there exists a hyperplane that separates the $1$-center from the points, and by moving the $1$-center towards that hyperplane, we are reducing its distance to each point). Therefore, the uncertainty region for the $1$-center is always within the convex hull of the uncertainty regions of the objects. We deduce that at any time, the uncertainty region for the $1$-center of all objects other than $B$ is within a strip-like region inside the convex hull of the uncertainty regions of the objects (see Figure~\ref{fig:1crd}). At any time, the length of this strip is $O(v  n)$ because the size of the uncertainty region of each object is $O(v  n)$ and the centers of the uncertainty regions of the objects are collinear. Moreover, the width of this strip is at most $r_q + 3v - (r_q - 9v) = 12v$. Hence, the maximum size of the strip-like region is $O(v  n)$. Consequently, using the above algorithm, the maximum size of the uncertainty region for the $1$-center of all objects other than $B$ is $O(v  n)$.

\begin{figure}[ht]
\centering
\begin{subfigure}{.45\textwidth}
    \centering
    \includegraphics[width=.85\linewidth]{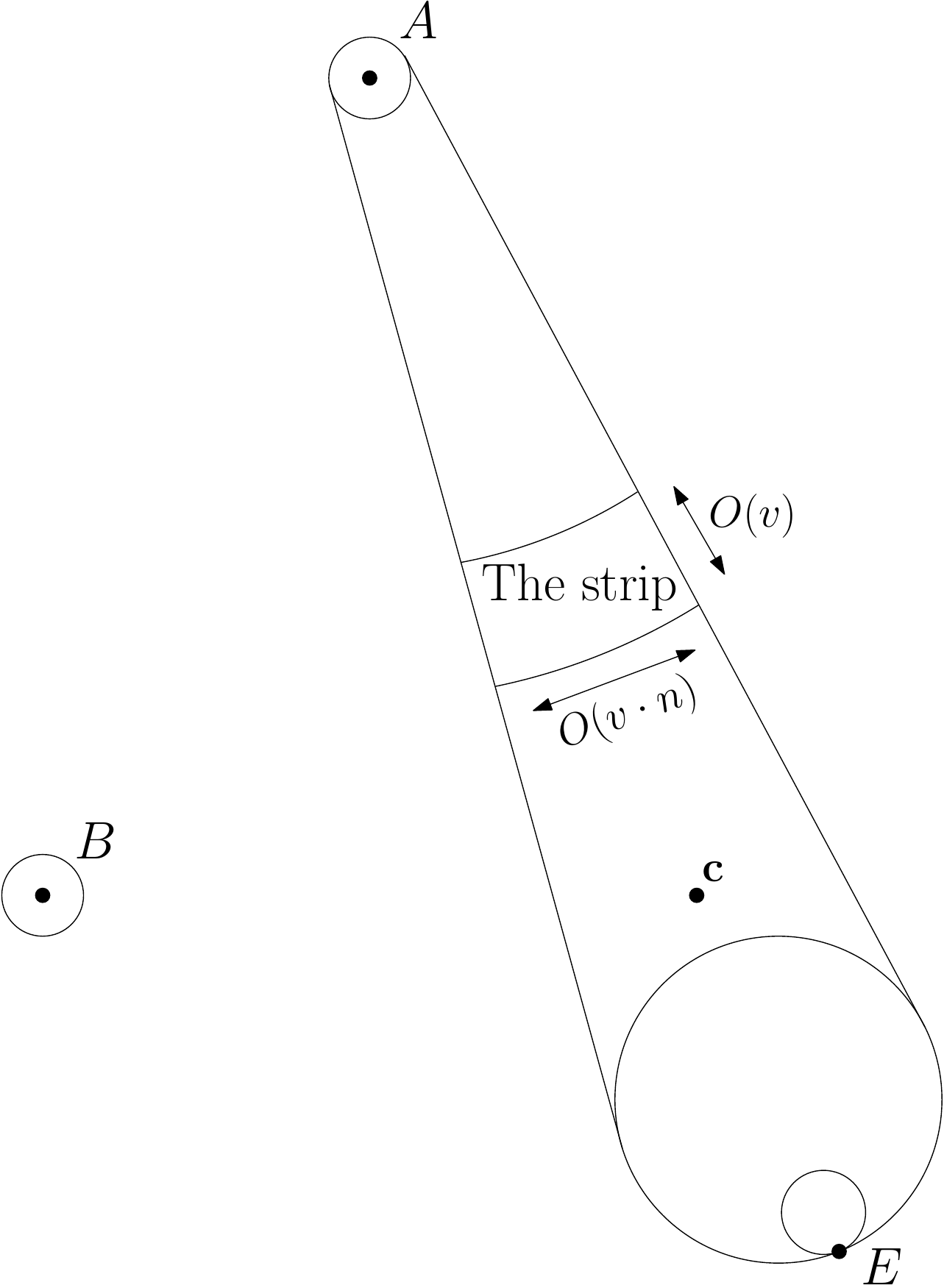}
    \caption{}
    \label{fig:1crd1}
\end{subfigure}
\begin{subfigure}{.45\textwidth}
    \centering
    \includegraphics[width=.85\linewidth]{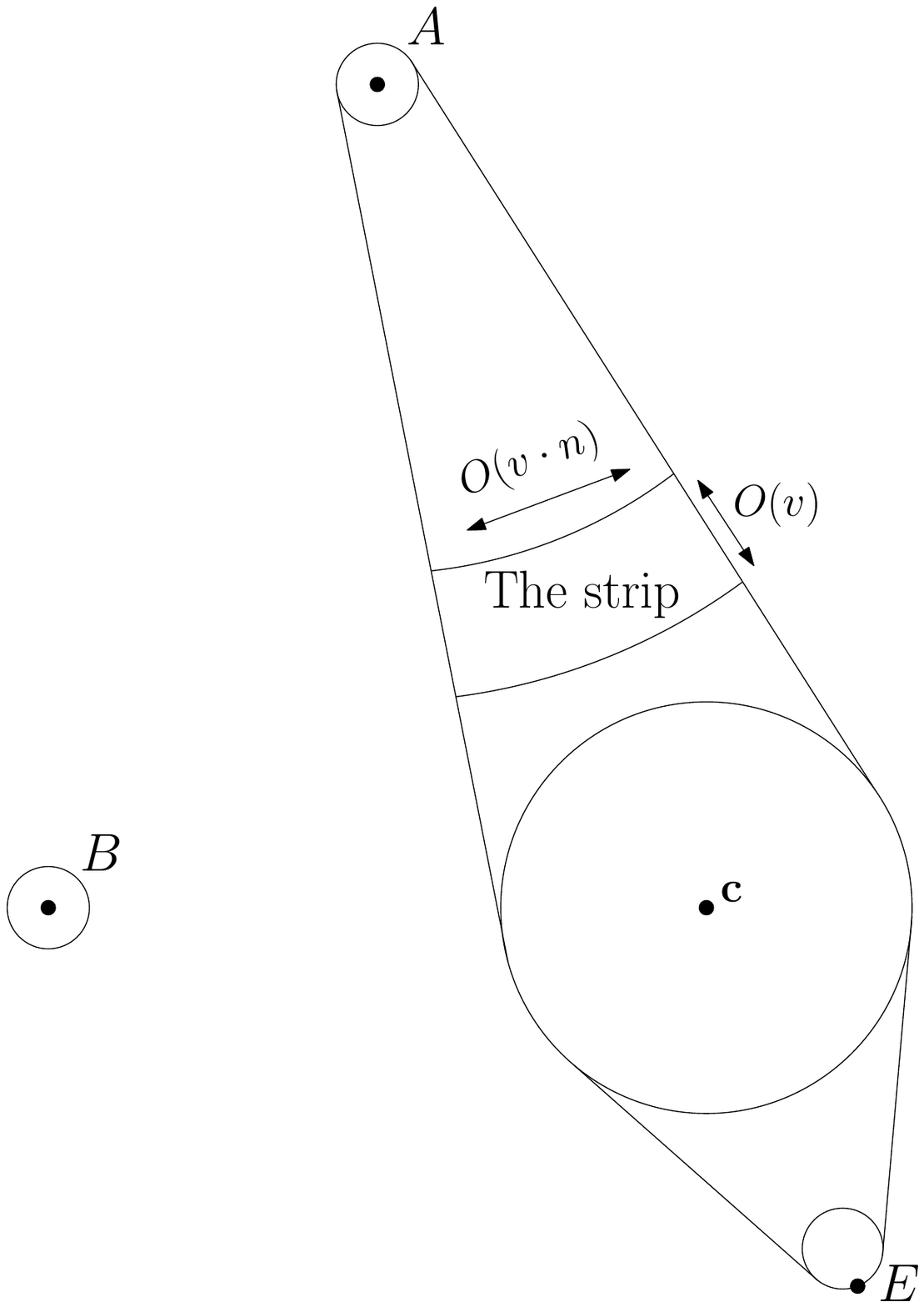}
    \caption{}
    \label{fig:1crd2}
\end{subfigure}
\caption{The uncertainty regions of the objects and the strip-like region that contains the $1$-center of all objects other than $B$ for the example and the algorithm provided in the proof of Theorem~\ref{thm:1crd} at the time (a) before the $(n-3)$-th query and (b) after the $(5n-15)$-th query.}
\label{fig:1crd}
\end{figure}

Let $\mathbf{d}$ represent the location of a point within the strip-like region. If the y-coordinate of $\mathbf{d}$ is within $[\frac{G}{3} + 3v, \frac{G}{2} - 4v]$, then the distance between $\mathbf{d}$ and any point in the uncertainty region $A$ is at least $\frac{G}{2} + v$. Also, in this case, the distance between $\mathbf{d}$ and any point in the uncertainty region $B$ is at most $\frac{G}{2}$. Furthermore, it is easy to see that if the y-coordinate of $\mathbf{d}$ is less than $\frac{G}{3} + 3v$, then the distance between $\mathbf{d}$ and any point in the uncertainty region of $B$ is less than the distance between $\mathbf{d}$ and any point in the uncertainty region of $A$.
If $r_q - 9v$ (the minimum distance between $\mathbf{a}$ and any point in the strip-like region) is at least $\frac{G}{2} + 5v$, then the y-coordinate of any point in the strip-like region is at most $\frac{G}{2} - 4v$, which means that 
the $1$-center of all objects other than $B$ is exactly the same as the $1$-center of all objects including $B$.
Considering that $n \geq 28$, we have
\[r_q - 9v = \frac{1}{2} |\mathbf{c} - \mathbf{a}| + \frac{v}{2} (n + q) - 9v \geq \frac{G}{2} + \frac{v}{2} (28 + 0) - 9v = \frac{G}{2} + 5v.\]
Hence, using the algorithm described above, the uncertainty region for the $1$-center of all objects including $B$ is always within the strip-like region.
Thus, the optimal algorithm keeps the maximum size of the uncertainty region for the $1$-center within $O(v  n)$ for the provided example.
\qed
\end{proof}
\section{The Centroid Problem}
\label{ch:centerofmass}

We know that at any time, the uncertainty region of each object is a ball centered at the last known exact position of the object.
We introduce a lemma which we will use to obtain the size of the uncertainty region for the center of mass or its special case, the centroid.
The proof of the lemma is provided in Appendix~\ref{app:CoM}.

\begin{lemma} \label{lem:CoM}
Given a set of $n$ weighted objects in $\mathbb{R}^d$ ($d \geq 1$) with uncertain locations, the uncertainty region for the center of mass is a ball centered at the weighted average of the centers of the objects' uncertainty regions. The radius of this ball is equivalent to the weighted average of the radii of the objects' uncertainty regions.
\end{lemma}

Lemma~\ref{lem:CoM} shows that the size of the uncertainty region for the center of mass is equivalent to the weighted average of the sizes of the objects' uncertainty regions.
Now, for the centroid problem, we prove that the Round-robin algorithm is $1$-competitive against the optimal algorithm. The proof of the following theorem is provided in Appendix~\ref{app:cent}.

\begin{theorem} \label{thm:cent}
For the centroid problem in $\mathbb{R}^d$ ($d \geq 1$), the maximum size of the uncertainty region for the Round-robin scheduling algorithm is $\Theta(v  n)$ and the Round-robin algorithm is $1$-competitive against the optimal algorithm.
\end{theorem}

For the center of mass problem, the competitive ratio of the Round-robin algorithm is $\Omega(n)$ against the optimal algorithm.
This can be easily proved by considering an example where the weight of one of the objects is much higher than the weight of the other objects.
Next, we provide an algorithm for the center of mass problem where objects have different maximum speeds and prove that it guarantees a competitive ratio of $O(\log n)$ against the optimal algorithm.

\begin{theorem}
For the center of mass problem in $\mathbb{R}^d$ ($d \geq 1$) where objects have different maximum speeds, there exists an algorithm that achieves a competitive ratio of $O(\log{n})$ against the optimal algorithm.
\end{theorem}

\begin{proof}
We sort the objects in a way that we have $v_1 w_1 \geq ... \geq v_n w_n$.
We then partition the objects into $\lfloor \log{n} \rfloor + 1$ groups. For $1 \leq i \leq n$, we assign the $i$-th object to the $(\lfloor \log{i} \rfloor + 1)$-th group.
Now, we describe the query strategy.
We query objects from different groups in a cyclic way.
Also, within each group, we choose the objects to query in a cyclic way.
For example, if we have seven objects, we partition the objects into three groups. We assign the first object to the first group, the second and third objects to the second group, and the remaining objects to the third group.
Consequently, our algorithm repeatedly follows the sequence $1, 2, 4, 1, 3, 5, 1, 2, 6, 1, 3, 7$ for querying the objects (entries of the sequence indicate the objects that our algorithm queries).

The algorithm described above acquires the exact location of one object from the group that contains the $i$-th object ($1 \leq i \leq n$) in every $\lfloor \log{n} \rfloor + 1$ queries and that group contains at most $2^{\lfloor \log{i} \rfloor}$ objects.
Therefore, the algorithm acquires the exact location of the $i$-th object at least once in every $(\lfloor \log{n} \rfloor + 1)  2^{\lfloor \log{i} \rfloor}$ queries.
Hence, using this algorithm, the maximum size of the uncertainty region for the center of mass is at most
$\frac{1}{W} \sum_{i=1}^{n}{2 ((\lfloor \log{n} \rfloor + 1)  2^{\lfloor \log{i} \rfloor} - 1)  v_i  w_i}$
(by Lemma~\ref{lem:CoM}), where $W = \sum_{i=1}^{n}{w_i}$ is the total weight of all objects.

On the other hand, to keep the maximum size of the uncertainty region for the center of mass finite, the optimal algorithm has to query the $n$-th object at some point. At that time, the size of the uncertainty region for the center of mass will be at least
$\frac{1}{W} \sum_{i=1}^{n-1}{2 i  v_i  w_i}$
(by Lemma~\ref{lem:CoM}).
We have
\begin{align*}
    \frac{3(\lfloor \log{n} \rfloor + 1)}{W} \sum_{i=1}^{n-1}{2 i  v_i  w_i} \geq & \ \frac{\lfloor \log{n} \rfloor + 1}{W} \sum_{i=1}^{n}{2 i  v_i  w_i} \\
    \geq & \ \frac{1}{W} \sum_{i=1}^{n}{2 ((\lfloor \log{n} \rfloor + 1)  i - 1)  v_i  w_i} \\
    \geq & \ \frac{1}{W} \sum_{i=1}^{n}{2 ((\lfloor \log{n} \rfloor + 1)  2^{\lfloor \log{i} \rfloor} - 1)  v_i  w_i}.
 \end{align*}
Therefore, for the center of mass problem, the competitive ratio of the algorithm described above is $O(\log{n})$ against the optimal algorithm.
\qed
\end{proof}
\section{The 1-Median Problem}
\label{ch:1median}
Given a set of $n$ points in $\mathbb{R}^1$, 
the $1$-median is the $\lceil \frac{n}{2} \rceil$-th smallest point if
$n$ is odd and the midpoint of the $\frac{n}{2}$-th and
$(\frac{n}{2}+1)$-th smallest points if $n$ is even.

\begin{restatable}{theorem}{TheoremOneMedMain}
\label{thm:1medmain}
For the $1$-median problem in $\mathbb{R}^1$, the Round-robin scheduling algorithm keeps the maximum size of the uncertainty region within $O(v  n)$ at any point in the future and achieves the best possible competitive ratio of $O(n)$ against the optimal algorithm.
\end{restatable}

\begin{proof} (sketch)
Let $s_k$ ($e_k$) be the $k$-th smallest starting
(ending) point of all $n$ objects' uncertainty regions.
Even though the region $[s_k,e_k]$ may not be the uncertainty region
of any particular object,
the $k$-th smallest object lies within this region,
which has size $O(v n)$, since Round-robin keeps the size of every
object's uncertainty region within $O(v  n)$.
The rest of the theorem follows from a lower bound of $\Omega(v)$ on the size of
$[s_k,e_k]$ 
and examples that show $O(n)$ is the best possible competitive ratio.
(See Appendix~\ref{app:1medmain}).
\qed
\end{proof}

The Round-robin algorithm keeps the maximum size of the uncertainty region for the $1$-median within $O(v  n)$. We argue that even for the optimal algorithm, the maximum size of the uncertainty region can be $\Omega(v  n)$.
The proof is presented in Appendix~\ref{app:1m1dOmega}.

\begin{proposition} \label{prop:1m1dOmega}
For the $1$-median problem in $\mathbb{R}^1$, the maximum size of the uncertainty region can be $\Omega(v  n)$ for the optimal algorithm.
\end{proposition}

Next, we provide an upper bound on the competitive ratio of the Round-robin algorithm for the $1$-median problem where objects have different maximum speeds.
The proof is presented in Appendix~\ref{app:1mdiffspeed}.

\begin{proposition} \label{prop:1mdiffspeed}
For the $1$-median problem in $\mathbb{R}^1$ where objects have different maximum speeds, the Round-robin scheduling algorithm achieves a competitive ratio of $O(v_M  n / v_m)$ against the optimal algorithm, where $v_M$ is the highest maximum speed and $v_m$ is the lowest maximum speed.
\end{proposition}

Now, we consider the $1$-median problem in $\mathbb{R}^d$ ($d \geq 2$).
The proof of the following proposition is inspired by the proof of Theorem 5.1 of \cite{durocher2006geometric} and can be found in Appendix~\ref{app:1msizeunb}.

\begin{proposition} \label{prop:1msizeunb}
For the $1$-median problem in $\mathbb{R}^d$ ($d \geq 2$), the maximum size of the uncertainty region can be unbounded for the optimal algorithm.
\end{proposition}

\section{Conclusion}
\label{ch:conclusion}

We conclude that in the worst case, adjusting query strategies based on the answers to the previous queries and the perceived locations of the objects does not help in achieving better competitive algorithms for any of the center problems discussed in this paper.

Here we list a handful of interesting problems that remain open:
\begin{itemize}
    \item For the static case of the $1$-center problem in $\mathbb{R}^1$, does there exist an algorithm that queries less than four objects per unit time and is $1$-competitive against the optimal algorithm with one query per unit time, or an algorithm that queries one object per unit time and achieves a competitive ratio of $o(n)$?
    \item For the $1$-center ($1$-median) problem in $\mathbb{R}^1$ where objects have different maximum speeds, can we find an algorithm with a competitive ratio of $o(v_M  n / v_m)$, where $v_M$ is the highest maximum speed and $v_m$ is the lowest maximum speed?
    \item For the $1$-median problem in $\mathbb{R}^d$ ($d \geq 2$), is there an algorithm that guarantees a bounded competitive ratio against the optimal algorithm?
\end{itemize}

\bibliographystyle{plain}
\bibliography{references}

\appendix
\section{Proof of Lemma~\ref{lem:sched}}\label{app:sched}

\begin{proof}
First, we show that given a multiset of positive integers $B$ where $d(B) \leq \frac{4}{3}$, it is possible to partition $B$ into two subsets $B_1$ and $B_2$ such that each subset has a pinwheel schedule. If $B$ contains one element equivalent to $1$, then we assign that element to $B_1$ and the rest of the elements to $B_2$. Else, if $B$ contains at least two elements in $\{ 2, 3 \}$, then we assign two of those elements to $B_1$ and the rest of the elements to $B_2$. Otherwise, if $B$ contains at least four elements in $\{ 4, 5, 6 \}$, then we assign four of those elements to $B_1$ and the rest to $B_2$. In all of the aforementioned cases, $B_1$ is schedulable with the Round-robin algorithm and $B_2$ is also schedulable because $d(B_2) = d(B) - d(B_1) \leq d(B) - \frac{2}{3} \leq \frac{2}{3}$. If none of the previous cases holds, then we assign every element of $B$ in $\{ 2, 3 \}$ to $B_1$ and every element of $B$ in $\{ 4, 5, 6 \}$ to $B_2$. The density of each subset would be less than or equal to $0.75$. Then, while $d(B_1) < 0.6$, we add the unassigned elements of $B$ to $B_1$, one by one. Each unassigned element is greater than or equal to $7$ and thus the amount that it adds to $d(B_1)$ is less than $0.15$. Hence, if $d(B_1) < 0.6$, the density of $B_1$ will not go above $0.75$ by adding one of the unassigned elements to the subset. If we can add all of the unassigned elements to $B_1$ while maintaining $d(B_1) < 0.6$, both $B_1$ and $B_2$ will be schedulable. If at any point, $d(B_1)$ becomes greater than or equal to $0.6$, we add the rest of the unassigned elements to $B_2$. This way, we have $d(B_1) \leq 0.75$ and $d(B_2) = d(B) - d(B_1) \leq d(B) - 0.6 \leq 0.75$;
thereby, both subsets are schedulable.

Now, we prove that if $d(A) \leq 2$, we can partition $A$ into three schedulable subsets $A_1$, $A_2$, and $A_3$. If $A$ contains either one element equivalent to $1$, two or more elements in $\{ 2, 3 \}$, four or more elements in $\{ 4, 5, 6 \}$, or five or more elements equivalent to $7$, we can choose $A_1$ such that it is schedulable and its density is greater than or equal to $\frac{2}{3}$. This way, we need to partition $A - A_1$ into two schedulable subsets. We know this is possible because $d(A - A_1) \leq \frac{4}{3}$. If none of the previous cases holds, then we assign every element of $A$ in $\{ 2, 3 \}$ to $A_1$, every element of $A$ in $\{ 4, 5, 6 \}$ to $A_2$, and every element equivalent to $7$ to $A_3$. The density of each subset would be less than or equal to $0.75$. Then, one by one, we add the unassigned elements of $A$ to the subset with the lowest density at the moment. Each unassigned element is greater than or equal to $8$. Since $d(A) \leq 2$, the density of the subset with the lowest density before adding an unassigned element with value $k$ is less than or equal to $\frac {2 - \frac{1}{k}}{3}$. Therefore, after adding the element, the density of the subset will be less than or equal to $\frac {2 - \frac{1}{k}}{3} + \frac{1}{k} = \frac{2}{3}  (1 + \frac{1}{k}) \leq \frac{2}{3}  (1 + \frac{1}{8}) = 0.75$. Consequently, we can add all of the unassigned elements to the three subsets without letting their densities exceed $0.75$; thereby, all three subsets will be schedulable.
\qed
\end{proof}

\section{Proof of Proposition~\ref{prop:1c1dOmega}}\label{app:1c1dOmega}

\begin{proof}
We consider an example with $n$ static objects located at the origin. At any point after at least $n-1$ queries, there will always be at least one object that the optimal algorithm will have not acquired its exact location in the last $n-1$ queries up to that point. The size of the uncertainty region of that object will be at least $2  v  (n-1)$. Furthermore, the size of the uncertainty region for the maximum and the minimum will be at least $v  (n-1)$. Consequently, the maximum size of the uncertainty region for the $1$-center will be $\Omega(v  n)$ for the optimal algorithm.
\qed
\end{proof}

\section{Proof of Proposition~\ref{prop:1cdiffspeed}}\label{app:1cdiffspeed}

\begin{proof}
The Round-robin algorithm keeps the maximum size of the uncertainty region for the $1$-center within $O(v_M  n)$.
On the other hand, the maximum size of the uncertainty region for the $1$-center is $\Omega(v_m)$ for the optimal algorithm.
Hence, the competitive ratio of the Round-robin algorithm for the $1$-center problem is $O(v_M  n / v_m)$.
\qed
\end{proof}

\section{Proof of Proposition~\ref{prop:1csizeunb}}\label{app:1csizeunb}

\begin{proof}
Consider an example with three static objects located at $(-2L, 0)$, $(2L, 0)$, and $(0, G)$, where $L$ is arbitrarily large and $G$ is much larger than $L$ in a way that the line segment between $(0, G)$ and $(2L, -v)$ is a diameter of a circle that encloses $(-2L, 0)$. Figure~\ref{fig:unbounded1center} illustrates this example.
If the object located at $(2L, 0)$ is not queried first, its uncertainty region after the first query contains $(2L, 0)$ and $(2L, -v)$. As a result, the uncertainty regions of the objects after the first query suggest that it is possible for the $1$-center to be located on the line $x = 0$ or the line $x = L$.
Thus, the size of the uncertainty region for the $1$-center will be at least $L$.
Similarly, if the object located at $(-2L, 0)$ is not queried first, the size of the uncertainty region for the $1$-center will be at least $L$.
We deduce that regardless of which object is queried first, the size of the uncertainty region for the $1$-center will be at least $L$ which is arbitrarily large.
Consequently, the maximum size of the uncertainty region for the $1$-center in $\mathbb{R}^d$ can be unbounded for the optimal algorithm.
\qed

\begin{figure}[ht]
    \centering
    \includegraphics[width=7cm]{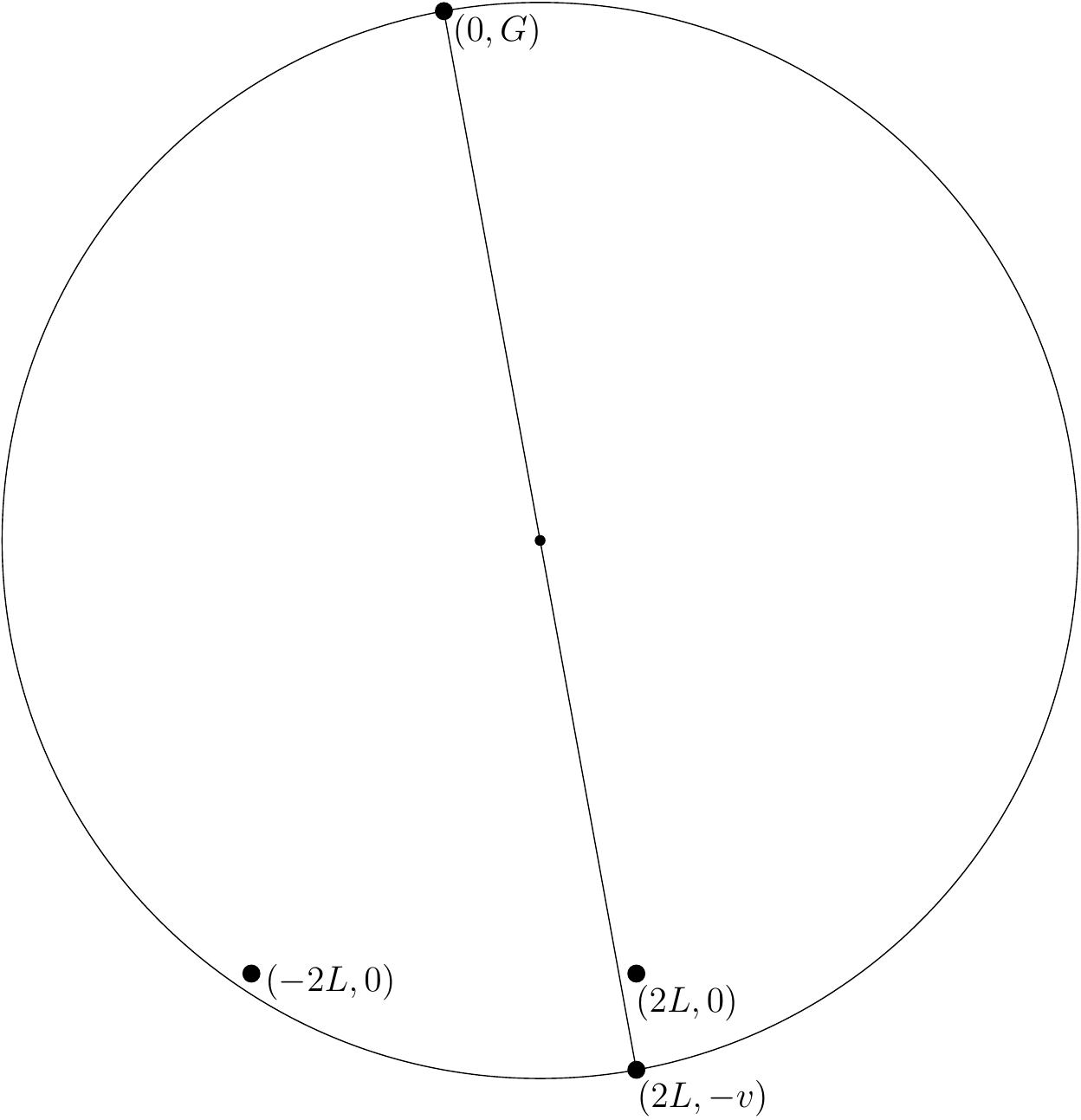}
    \caption{Illustration for the example provided in the proof of Proposition~\ref{prop:1csizeunb}.}
    \label{fig:unbounded1center}
\end{figure}

\end{proof}

\section{Proof of Lemma~\ref{lem:CoM}}\label{app:CoM}

\begin{proof}
For $1 \leq i \leq n$, let $\mathbf{c}_i$ and $r_i$ denote the center and the radius of the uncertainty region of the $i$-th object, respectively. Also, let $w_i$ be the weight of the $i$-th object and $W = \sum_{i=1}^{n}{w_i}$ be the total weight of all objects. We want to prove that the uncertainty region for the center of mass is a ball centered at $\mathbf{C} = \frac{1}{W} \sum_{i=1}^{n}{w_i \mathbf{c}_i}$ with radius $R = \frac{1}{W} \sum_{i=1}^{n}{w_i r_i}$.

Let $\mathbf{q}_1, ..., \mathbf{q}_n$ represent the locations of $n$ arbitrary points such that for $1 \leq i \leq n$, the $i$-th point is within the uncertainty region of the $i$-th object. We show that the weighted average of these points is within a ball centered at $\mathbf{C}$ with radius $R$. We have
\[
    |\frac{1}{W} \sum_{i=1}^{n}{w_i \mathbf{q}_i} - \mathbf{C}| = |\frac{1}{W} \sum_{i=1}^{n}{w_i (\mathbf{q}_i - \mathbf{c}_i)}|
    \leq \frac{1}{W} \sum_{i=1}^{n}{w_i |\mathbf{q}_i - \mathbf{c}_i|}
    \leq \frac{1}{W} \sum_{i=1}^{n}{w_i r_i}
    = R.
\]

Now, let $\mathbf{U}$ represent the location of an arbitrary point within a ball centered at $\mathbf{C}$ with radius $R$. We show that $\mathbf{U}$ is the weighted average of $n$ points $\mathbf{u}_1, ..., \mathbf{u}_n$ such that for $1 \leq i \leq n$, $\mathbf{u}_i = \mathbf{c}_i + \frac{r_i}{R} (\mathbf{U} - \mathbf{C}) $ and hence, the $i$-th point is within the uncertainty region of the $i$-th object. We have
\[
\frac{1}{W} \sum_{i=1}^{n}{w_i \mathbf{u}_i}
=
\frac{1}{W} \sum_{i=1}^{n}{w_i (\mathbf{c}_i + \frac{r_i}{R}(\mathbf{U} - \mathbf{C}))}
=
\mathbf{C} + \left(\frac{1}{WR} \sum_{i=1}^{n}{w_i r_i}\right)(\mathbf{U} - \mathbf{C})
=
\mathbf{U} .
\]
Therefore, the uncertainty region for the center of mass is a ball centered at $\mathbf{C}$ with radius $R$.
\qed
\end{proof}

\section{Proof of Theorem~\ref{thm:cent}}\label{app:cent}

\begin{proof}
At any time after the Round-robin algorithm makes $n$ queries, the size of the uncertainty region of the object queried $i$-th to last ($2 \leq i \leq n$) is $2  v  (i - 1)$.
Thus, by Lemma~\ref{lem:CoM}, the size of the
uncertainty region for the centroid
is
$\frac{1}{n}  \sum_{i=1}^{n}{2  v  (i - 1)} = v  (n-1)$.
Hence, the maximum size of the centroid’s uncertainty region for the Round-robin algorithm is $\Theta(v  n)$.

For the optimal algorithm, at any point after at least 
$n$
queries, the size of the $i$-th smallest uncertainty region ($1 \leq i \leq n$) among the uncertainty regions of all objects is at least $2  v  (i - 1)$. Consequently, by Lemma~\ref{lem:CoM}, the size of the uncertainty region for the centroid is at least $v  (n-1)$. Therefore, the Round-robin algorithm is $1$-competitive against the optimal algorithm. Furthermore, the maximum size of the centroid’s uncertainty region for the optimal algorithm is $\Theta(v  n)$.
\qed
\end{proof}

\section{Proof of Theorem~\ref{thm:1medmain}}\label{app:1medmain}

\begin{lemma} \label{lem:klow}
 The maximum size of the uncertainty region of the $k$-th smallest ($2 \leq k \leq n-1$), over all time is $\Omega(v)$, regardless of the query strategy.
\end{lemma}

\begin{proof}
We prove this by contradiction.
We assume that there exists an example where using an algorithm $A$, the size of the uncertainty region for the $k$-th smallest will never exceed $\frac{v}{2}$.
For $1 \leq i \leq n$ and $q \geq 0$, let $s_{i, q}$ ($e_{i, q}$) be equal to the $i$-th smallest starting (ending) point among the starting (ending) points of the uncertainty regions of all objects after exactly $q$ queries are made by $A$.
We say that $s_{i, q}$ ($e_{i, q}$) bounds an object $b$ if it is equal to the starting (ending) point of the uncertainty region of $b$.
$[s_{i, q}, e_{i, q}]$ is the uncertainty region for the $i$-th smallest after $q$ queries. 
By our assumption, $e_{k, q}-s_{k, q} \leq \frac{v}{2}$ for all $q$.


Let us consider the case that for some $q \geq 0$, we have $e_{k, q} - e_{k-1, q} \leq \frac{v}{2}$. If $s_{k+1, q} \leq e_{k, q}$, after the $(q+1)$-th query, we will have at least $k$ starting points less than or equal to $e_{k, q} - v$. Also, we will have $e_{k, q+1} \geq e_{k-1, q} + v$. Thus, we will have $s_{k, q+1} \leq e_{k, q} - v \leq e_{k-1, q} - \frac{v}{2} \leq e_{k, q+1} - \frac{3v}{2}$. However, this is against the assumption that using algorithm $A$, the size of the uncertainty region for the $k$-th smallest will never exceed $\frac{v}{2}$. Consequently, we must have $s_{k+1, q} > e_{k, q}$. If the $(q+1)$-th query acquires the exact location of the object that was bounded by $e_{k, q}$ after $q$ queries, the new position of that object must be greater than or equal to $e_{k-1, q} + \frac{v}{2}$ so that we have $e_{k, q+1}-s_{k, q+1} \leq \frac{v}{2}$. In that case, we have $e_{k, q+1} \geq e_{k, q} + \frac{v}{2}$ and $e_{k, q+1} - e_{k-1, q+1} \leq \frac{v}{2}$. If the $(q+1)$-th query acquires the exact location of any other object, the starting point of the uncertainty region of that object before the $(q+1)$-th query must have been less than or equal to $s_{k, q}$ (this is a necessary condition for having $e_{k, q+1}-s_{k, q+1} \leq \frac{v}{2}$). Hence, in that case, $e_{k, q+1} = e_{k, q} + v$. Moreover, the new position of the queried object must be greater than or equal to $e_{k, q} + \frac{v}{2}$ so that we have $e_{k, q+1}-s_{k, q+1} \leq \frac{v}{2}$. For this reason, we have $e_{k, q+1} - e_{k-1, q+1} \leq \frac{v}{2}$. Therefore, in both cases, we have $e_{k, q+1} - e_{k-1, q+1} \leq \frac{v}{2}$, $e_{k, q+1} \geq e_{k, q} + \frac{v}{2}$ and $s_{k+1, q+1} < s_{k+1, q}$. We deduce that for some $q' > q$, we will have $e_{k, q'} - e_{k-1, q'} \leq \frac{v}{2}$ and $s_{k+1, q'} \leq e_{k, q'}$. As a result, the size of the uncertainty region for the $k$-th smallest will exceed $\frac{v}{2}$. Hence, if we have $e_{k, q} - e_{k-1, q} \leq \frac{v}{2}$ for some $q \geq 0$, the size of the uncertainty region for the $k$-th smallest will exceed $\frac{v}{2}$ which is against the assumption about algorithm $A$. Similarly, we can show that if we have $s_{k+1, q} - s_{k, q} \leq \frac{v}{2}$ for some $q \geq 0$, the size of the uncertainty region for the $k$-th smallest will exceed $\frac{v}{2}$ at some point.

We now consider the case that for some $q \geq 0$, $s_{k, q}$ and $e_{k, q}$ bound two different objects. If the object bounded by $e_{k, q}$ was not queried last, then at most $k-2$ ending points can be smaller than $s_{k, q}$; thereby, $s_{k, q} \leq e_{k-1, q}$. Therefore, we have $e_{k, q} - e_{k-1, q} \leq e_{k, q} - s_{k, q} \leq \frac{v}{2}$. If the object bounded by $e_{k, q}$ was queried last, then we have a starting point at $e_{k, q}$. In that case, we have $s_{k+1, q} - s_{k, q} \leq e_{k, q} - s_{k, q} \leq \frac{v}{2}$. We deduce that either $e_{k, q} - e_{k-1, q} \leq \frac{v}{2}$ or $s_{k+1, q} - s_{k, q} \leq \frac{v}{2}$. As a result, the size of the uncertainty region for the $k$-th smallest will exceed $\frac{v}{2}$ at some point which is against the assumption about algorithm $A$.

Hence, for every $q \geq 0$, $s_{k, q}$ and $e_{k, q}$ must bound the same object. Furthermore, we must have $e_{k, q} - e_{k-1, q} > \frac{v}{2}$ and $s_{k+1, q} - s_{k, q} > \frac{v}{2}$. Since we assume that $e_{k, q}-s_{k, q} \leq \frac{v}{2}$, we must have the exact location of the object bounded by $s_{k, q}$ and $e_{k, q}$ after $q$ queries. This means that $e_{k, q}-s_{k, q} = 0$ and if $q > 0$, the $q$-th query must have acquired the exact location of the object bounded by $s_{k, q}$ and $e_{k, q}$. If after $q$ queries ($q \geq 0$), we have the exact location of an object $t$, we know that $t$ has the $k$-th smallest position, and we also have $e_{k, q} - e_{k-1, q} > \frac{v}{2}$ and $s_{k+1, q} - s_{k, q} > \frac{v}{2}$, $t$ can potentially have the $k$-th smallest position after the $(q+1)$-th query, regardless of which object will be queried. Thus, no other object $t'$ exists such that we have its exact location after $q+1$ queries and that we know with certainty that it has the $k$-th smallest position. Therefore, if at any point, algorithm $A$ queries an object whose initial position was not the $k$-th smallest, the size of the uncertainty region for the $k$-th smallest will eventually exceed $\frac{v}{2}$. On the other hand, if algorithm $A$ always queries the same object $t$ whose initial position was the $k$-th smallest, then the uncertainty regions of the other objects will become extensive. In that case, considering that the initial position of $t$ was the $k$-th smallest ($2 \leq k \leq n-1$), the uncertainty region of at least one other object will include the position of $t$ at some point in the future. When this happens, the position of $t$ will not necessarily be the $k$-th smallest. Consequently, the size of the uncertainty region for the $k$-th smallest will exceed $\frac{v}{2}$ at some point.

In every case, we have proved that using algorithm $A$, the size of the uncertainty region for the $k$-th smallest will exceed $\frac{v}{2}$. This contradicts our assumption about algorithm $A$. Therefore, for every algorithm, the maximum size of the uncertainty region for the $k$-th smallest is greater than $\frac{v}{2}$. Thus, when $2 \leq k \leq n-1$, regardless of the query strategy, the maximum size of the uncertainty region for the $k$-th smallest is $\Omega(v)$.
\qed
\end{proof}

\TheoremOneMedMain*

\begin{proof}
When $n$ is odd, the $1$-median problem is the same as the $\lceil \frac{n}{2} \rceil$-th smallest problem. Moreover, when $n$ is even, we consider the $1$-median to be the midpoint of the $\frac{n}{2}$-th smallest and $(\frac{n}{2}+1)$-th smallest points. In this case, the size of the uncertainty region for the $1$-median is equivalent to the average of the size of the uncertainty region for the $\frac{n}{2}$-th smallest and the size of the uncertainty region for the $(\frac{n}{2}+1)$-th smallest.

The Round-robin algorithm keeps the size of the uncertainty region of each object within $O(v  n)$ because it acquires the exact location of each object once in every $n$ queries.
For $1 \leq k \leq n$, let $s_k$ be the $k$-th smallest starting point among the starting points of the uncertainty regions of all objects.
Also, let $e_k$ be the $k$-th smallest ending point among the ending points of the uncertainty regions of all objects.
Hence, $[s_k, e_k]$ is the uncertainty region for the $k$-th smallest.
Using the Round-robin algorithm, we know that $e_k - s_k$ will always be $O(v  n)$ because the algorithm keeps the size of the uncertainty region of each object within $O(v  n)$.
Thus, the size of the uncertainty region for the $k$-th smallest is bounded by $O(v  n)$ at any point in the future.
Therefore, using the Round-robin algorithm, the size of the uncertainty region for the $1$-median will never exceed $O(v  n)$.

If $n=2$, at any point after the first query, the size of the uncertainty region of one of the objects will be $\Omega(v)$. Thus, the maximum size of the uncertainty region for the $1$-median will be $\Omega(v)$.
If $n \geq 3$, by Lemma~\ref{lem:klow}, the maximum size of the uncertainty region for the $k$-th smallest ($2 \leq k \leq n-1$) will be $\Omega(v)$, regardless of the query strategy. 
Therefore, for the optimal algorithm, the size of the uncertainty region for the $1$-median will be $\Omega(v)$ at some point in the future. Considering that the Round-robin algorithm keeps the maximum size of the uncertainty region within $O(v  n)$, it achieves a competitive ratio of $O(n)$ against the optimal algorithm.


Now, we show that $O(n)$ is the best possible competitive ratio against the optimal algorithm. We consider the following cases:
\begin{itemize}
    \item If $n$ is odd, we consider an example with $\lceil \frac{n}{2} \rceil$ objects initially located at the origin and $\lfloor \frac{n}{2} \rfloor$ objects initially located at $10  v  n$. One of the objects initially located at the origin moves with a speed of $v$ towards $3  v  n$ and then stays there. The rest of the objects do not move. Let us consider an algorithm that acquires the exact location of the moving object once in every two queries (even after the object stops at $3  v  n$) and assigns the rest of the queries to the static objects in a circular order. This algorithm maintains the size of the uncertainty region for the $\lceil \frac{n}{2} \rceil$-th smallest within $O(v)$.
    \item If $n$ is even, we consider an example with $\frac{n}{2}$ objects initially located at the origin and $\frac{n}{2}$ objects initially located at $10  v  n$. One of the objects initially located at the origin moves with a speed of $v$ towards $3  v  n$ and then stays there. Also, one of the objects initially located at $10  v  n$ moves with a speed of $v$ towards $7  v  n$ and then stays there. The rest of the objects do not move. We consider an algorithm that acquires the exact location of each of the two moving objects once in every three queries (even after the objects stop moving) and assigns the rest of the queries to the static objects in a circular order. This algorithm maintains the size of the uncertainty region for the $\frac{n}{2}$-th and $(\frac{n}{2}+1)$-th smallest within $O(v)$.
\end{itemize}
Therefore, in the examples provided for both cases, the optimal algorithm keeps the maximum size of the uncertainty region for the $1$-median within $O(v)$.
However, any algorithm that does not know the trajectories of the objects, such as Round-robin, may fail to query the moving object in the first $\lceil \frac{n}{2} \rceil - 1$ attempts;
hence, the size of the uncertainty region for the $\lceil \frac{n}{2} \rceil$-th smallest can become $\Omega(v  n)$ after the first $\lceil \frac{n}{2} \rceil - 1$ queries. For this reason, the size of the uncertainty region for the $1$-median can become $\Omega(v  n)$. This is true even if we allow the algorithm to query a constant number of objects (instead of one) per unit time. Thus, $O(n)$ is the best possible competitive ratio.
\qed
\end{proof}

\section{Proof of Proposition~\ref{prop:1m1dOmega}}\label{app:1m1dOmega}

\begin{proof}
We consider an example with $n$ static objects located at the origin. At any point after at least $\lfloor \frac{n}{2} \rfloor$ queries, there will always be at least $\lceil \frac{n}{2} \rceil$ objects that the optimal algorithm will have not acquired their exact locations in the last $\lfloor \frac{n}{2} \rfloor$ queries up to that point. The ending points of the uncertainty regions of those $\lceil \frac{n}{2} \rceil$ objects will be at least $v  \lfloor \frac{n}{2} \rfloor$. Thus, the $(\lfloor \frac{n}{2} \rfloor + 1)$-th smallest ending point among the ending points of the uncertainty regions of all objects will be at least $v  \lfloor \frac{n}{2} \rfloor$. On the other hand, the starting points of the uncertainty regions of all objects will always be less than or equal to zero. Therefore, the size of the uncertainty region for the $(\lfloor \frac{n}{2} \rfloor + 1)$-th smallest will be $\Omega(v  n)$ at any point after the first $\lfloor \frac{n}{2} \rfloor$ queries. For this reason, the maximum size of the uncertainty region for the $1$-median will be $\Omega(v  n)$ for the optimal algorithm.
\qed
\end{proof}

\section{Proof of Proposition~\ref{prop:1mdiffspeed}}\label{app:1mdiffspeed}

\begin{proof}
The Round-robin algorithm keeps the maximum size of the uncertainty region for the $1$-median within $O(v_M  n)$. On the other hand, the maximum size of the uncertainty region for the $1$-median is $\Omega(v_m)$ for the optimal algorithm. Thus, the competitive ratio of the Round-robin algorithm for the $1$-median problem is $O(v_M  n / v_m)$.
\qed
\end{proof}

\section{Proof of Proposition~\ref{prop:1msizeunb}}\label{app:1msizeunb}

\begin{proof}
Consider an example with four static objects located at $(0, 0)$, $(0, v)$, $(L, 0)$, and $(L, v)$, where $L$ is arbitrarily large. Regardless of which object is queried first, the uncertainty regions of the objects after the first query suggest that it is possible to have two objects with the exact same location on the line $x = 0$ and two objects with different locations on the line $x = L$. In this case, the $1$-median will be located on the line $x = 0$. Furthermore, the uncertainty regions of the objects after the first query suggest the possibility of having two objects with different locations on the line $x = 0$ and two objects with the exact same location on the line $x = L$. In this case, the $1$-median will be located on the line $x = L$. Therefore, regardless of which object is queried first, the size of the uncertainty region for the $1$-median will be at least $L$ which is arbitrarily large. Consequently, the maximum size of the uncertainty region for the $1$-median in $\mathbb{R}^d$ can be unbounded for the optimal algorithm.
\qed
\end{proof}

\end{document}